\documentclass[preprint,12pt]{elsarticle}

\usepackage{enumitem}
\usepackage{graphicx}
\usepackage{comment}

\usepackage{amsfonts}
\usepackage{amsthm}
\usepackage{amsmath}
\usepackage{amssymb}

\newtheorem{theorem}{Theorem}
\newtheorem{lemma}[theorem]{Lemma}

\newtheorem{corollary}[theorem]{Corollary}

\newcommand{\ceiling}[1]{\left\lceil #1 \right\rceil}

\begin{document}

\begin{frontmatter}

\title{Streaming Algorithms for \\ Multitasking Scheduling with Shared Processing
}

\author[inst1]{Bin Fu}

\affiliation[inst1]{organization = {Department of Computer Science},
addressline = {University of Texas Rio Grande Valley}, city = { Edinburg}, state = {TX}, postcode={78539}, country={USA}}

\author[inst2]{Yumei Huo}

\affiliation[inst2]{organization = {Department of Computer Science},
addressline = {College of Staten Island, CUNY}, city = { Staten Island}, state = {NY}, postcode={10314}, country={USA}}

\author[inst2]{Hairong Zhao}

\affiliation[inst3]{organization = {Department of Computer Science},
addressline = {Purdue University Northwest}, city = { Hammond}, state = {IN}, postcode={46323}, country={USA}}

\begin{abstract}

In this paper, we design the first streaming algorithms for the problem of multitasking scheduling on parallel machines with shared processing.
In one pass, our streaming approximation schemes can provide an approximate value of the optimal makespan. If the jobs can be read in two passes, the algorithm can find the schedule with the approximate value. This work not only provides an algorithmic big data solution for the studied problem, but also gives an insight into the design of   streaming algorithms for other problems in the area of scheduling.

\end{abstract}


\end{frontmatter}

\setlength{\baselineskip}{20pt}

\section{Introduction}\label{intro}
In recent years, with more and more data  generated in all applications, the dimension of the computation increases rapidly. As in many other research areas, the need for providing  solutions under big data also emerges in the area of scheduling.  In this paper,
we study the data stream model of multitasking scheduling problem. Under this data  model,  the input data is massive and cannot  be read into memory;  the goal is to design streaming algorithms to approximate the optimal solution in a few passes (typically just one) over the data and using limited space.

As Muthukrishnan addressed in his paper \cite{m05},
traditionally, data are fed from the memory and one can modify the underlying data to reflect the updates; real time queries are also simple by looking up a value in the memory, as we see in banking and credit transactions; as far as complex analyses are concerned, such as trend analysis, forecasting, etc., operations are usually performed offline. However, in the modern world, with more and more data generated in the monitoring applications such as atmospheric, astronomical, networking, financial, sensor-related fields, etc., the automatic data feeds are needed for many tasks. For example,
large amount of data need to be fed and processed in a short time to monitor complex correlations, track trends, support exploratory analyses and perform complex tasks such as classification, harmonic analysis etc. These tasks are time critical and thus it is important to process them in near-real time to accurately keep pace with the rate of stream updates and  reflect rapidly changing trends in the data.
With more data generated and more demands of data streams processing for now and in the future, the researchers are facing the questions: Given a certain amount of resources, a data stream rate and a particular analysis task, what can we (not) do?

While some methods are available for processing large amount of data of these time critical tasks, such as making things parallel, controlling data rate by sampling or shedding updates, rounding data structures to certain block boundaries, using hierarchically detailed analysis, etc., these approaches are ultimately limiting.

A natural approach to dealing with data streams involves approximations and developing   algorithmic principles for data stream algorithms. Streaming algorithms were initially studied by Munro and Paterson in 1978 (\cite{mp78}), and then by Flajolet and Martin in 1980s (\cite{fm85}). The model was formally established by Alon, Matias, and Szegedy in ~\cite{ams99} and has received a lot of attention since then.
Formally, streaming algorithms are algorithms for processing the input where some or all of of the data is not available for random access but rather arrives as a sequence of items and can be examined in only a few passes (typically just one).
The performance of streaming algorithms is measured by three factors: the number of passes the algorithm must run over the stream, the space needed and the updating time of the algorithm.

In this paper, we study the streaming algorithms for the problem of multitasking scheduling with shared processing. Over the past couple of decades, the problem of multitasking scheduling has attracted a lot of attention in the service industries where workers frequently perform multiple tasks by switching from one task to another.
For example, in health care, 21\% of hospital employees spend their working time on more than one activity \cite{olb06}, and in consulting where workers usually engage in about 12 working spheres per day \cite{gm05}.
Although in the literature some research has been done on the effect of multitasking (\cite{gm05}, \cite{vmn08}, \cite{cip14}, \cite{ssmw13}), the study on multitasking in the area of scheduling is still very limited (\cite{hll15}, \cite{hll16},
\cite{sh15}, \cite{zzc17}).

Hall, Leung and Li ( \cite{hll16} ) proposed a multitasking scheduling model that allows a team to continuously work on its main, or primary tasks while a fixed percentage of its processing capacity may be allocated to process the routinely scheduled activities as they occur.
Some examples of the routinely scheduled activities are administrative meetings, maintenance work, or meal breaks. In these scenarios, some team members need to be assigned to perform these routine activities while the remaining team members still focus on the primary tasks.
Since the routine activities are essential to the maintenance of the overall system in many situations, they are usually managed separately and scheduled independently of the primary jobs.
When these multitasking problems are modeled in the scheduling theory, a working team is viewed as a machine which may have some periods during which routine jobs and primary jobs share the processing capacity.

In \cite{hll16}, it is assumed that there is only a single machine and the machine capacity allocated to routine jobs  is the same for all routine jobs.
In this paper, we generalize this model to parallel machine environment and allow the machine capacity allocated to routine jobs to vary from one to another. In practice, it is not uncommon
that different  number of team members
are needed to perform different routine jobs during different time periods.

In many circumstances, it is necessary to have service continuously available for primary jobs, such as in many companies' customer service and technical support departments, the service must be continuous for answering customers' calls and for troubleshooting the customers' product failures.
So at least one member from the team is needed to provide these service at any time while the size of a team is typically of ten or fewer members as recommended by Dotdash Meredith Company in their management research. To model this, we allow  the capacities allocated for primary jobs on some machines to have a constant lower bound.

\subsection{Problem Definition}

Formally, our problem can be defined as follows. We are given $m$ identical machines $\{M_1, M_2, \ldots, M_m\}$ and a set $N = \{1, \ldots, n\}$ of primary jobs that are all available for processing at time 0. Each primary job $j \in N $ has a processing time $p_j$ and can be processed by any one of the machines uninterruptedly.
Each machine $M_i$ has $k_i$ shared processing intervals during which only a fraction of machine capacity can be allocated to these primary jobs due to the fact some capacity has been pre-allocated to routine jobs. We use ``sharing ratio'' to refer the fraction of the capacity allocated to the primary jobs. The total number of these intervals  is  $\tilde{n}  = \sum_{1 \le i \le m} k_i$.
For simplicity, we will treat those  intervals with full capacity as intervals with  sharing ratio 1.
 Apparently, each machine $M_i$ has  $O(k_i)$ intervals in total.
Without loss of generality, we assume that these intervals are given in  sorted order, denoted as  $I_{i,1} = (0, t_{i,1}]$, $I_{i,2} = (t_{i,1}, t_{i,2}]$, $\ldots$, and their corresponding sharing ratios are $e_{i,1}$, $e_{i,2}$, $\ldots$, all of which  are  in the range of $(0,1]$, see Figure~\ref{fig:sharing intervals}(a) for an illustration of machine intervals and Figure~\ref{fig:sharing intervals}(b) for an illustration of a schedule of primary jobs in these intervals.
For any schedule $S$, let $C_j(S)$ be the completion time of the primary job $j$ in $S$. If the context is clear, we use $C_j$ for short. The makespan of the schedule $S$ is $\max_{1\le j \le n} \{C_j\}$.
The objective is to find a schedule of the primary jobs to minimize the makespan.

\begin{figure}
 \includegraphics[width=\textwidth]{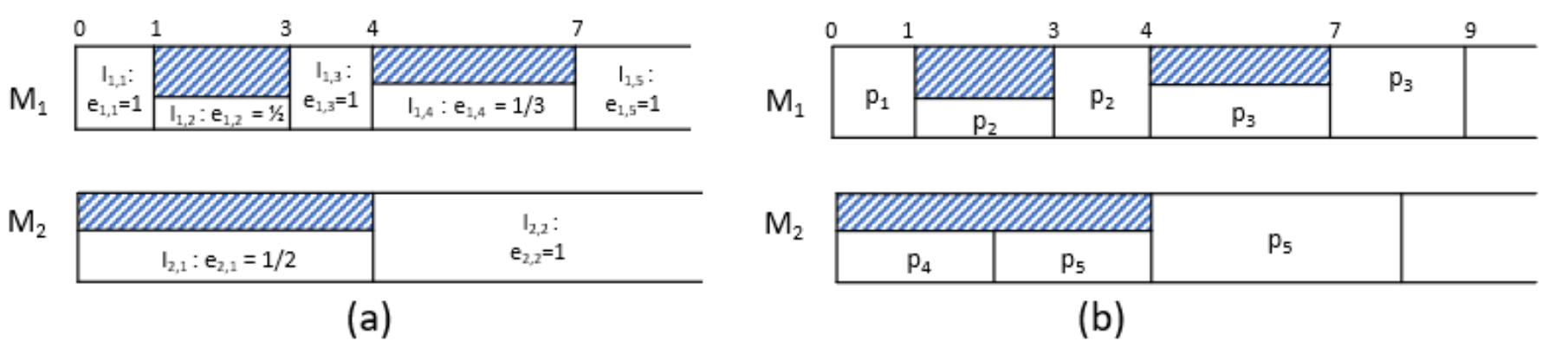}
\caption{An example of multitasking scheduling on 2 machines with shared processing, 3 routine jobs are shown as shaded intervals. (a) The intervals and the sharing ratios; (b) A schedule of  five primary jobs: $p_1 = 1$, $p_2 = 2$, $p_3 = 3$, $p_4 = 1$, $p_5 = 5$.  } \label{fig:sharing intervals}
\end{figure}

In this paper, we consider the above scheduling  problem under the data stream model.
Specifically, we study the problem that the number of primary jobs is so big that jobs' information cannot be stored in the memory but can only be scanned in one or more passes. Extending the three-field $\alpha \mid \beta \mid \gamma$ notation introduced by Graham et al. \cite{gllr79}, our problem is denoted as $P_m\mid \text{stream, share} \mid C_{max}$ if the sharing ratios are arbitrary;
and if the sharing ratio is at least $e_0$ for the intervals on the first $m_1$ ($ 1 \le m_1 \le m-1$) machines but arbitrary for other $(m-m_1)$ machines, our problem is denoted as $P_m\mid \text{stream, share, }  e_{i,k} \ge e_0 \text{ for } i \le m_1 \mid C_{max}$. The corresponding problems under the traditional data model will be denoted as $P_m\mid \text{share} \mid C_{max}$ and $P_m\mid \text{share, }  e_{i,k} \ge e_0 \text{ for } i \le m_1 \mid C_{max}$, respectively.

\subsection{Literature Review}

Various models of shared processing scheduling  have been studied in the literature. In the model studied by \cite{dk17}, \cite{hk15}, \cite{dk20},
jobs have their own private processors and can also be processed by other processors which are shared by other jobs to reduce the job’s completion time due to processing time overlap.
In the model studied by Baker and Nuttle \cite{bn80}, all the resources together are viewed as one machine which has varying availability over time and jobs are scheduled on this single machine with varying capacity.
The authors showed that a number of well-known results for classical single-machine problems can be applied with little or no modification to the corresponding variable-resource problems. Then Hirayama and Kijima \cite{hk92} studied this problem when the machine capacity varies stochastically over time. Adiri and Yehudai \cite{ay87} studied the problem on single and parallel machines such that  the service rate of a machine can only be changed when a job is completed.

The shared processing multitasking model studied in this paper was first proposed by Hall et. al. in \cite{hll16}. In this model, machine may have reduced capacity for processing primary jobs in some periods where routine jobs are scheduled and share the processing with primary jobs. They studied this model in the single machine environment and assumed that the sharing   ratio   is a constant $e$ for all the shared intervals. For this model, it is easy to see that the makespan is the same for all schedules that have no unnecessary idle time. The authors in  \cite{hll16}  showed that the total completion time can be minimized by scheduling the jobs in non-decreasing order of the processing time, but  it  is unary NP-Hard for the objective function of total weighted completion time. When the primary jobs have different due dates, the authors    gave  polynomial time algorithms for  maximum lateness and the number of late jobs.

For our studied problems, if $e_{i,k} = 1$ for all time intervals, that is, there are no routine jobs, the problem becomes the classical parallel machine scheduling problem $P_m\mid \mid C_{max}$. For this problem,  Graham   studied the performance of List Scheduling  rule (\cite{g66})
and  Longest Processing Time (\cite{g69})  rule.
When the number of machines $m$ is fixed, Horowitz and Sahni \cite{hs76} developed a fast approximation scheme. Later,  Hochbaum and Shmoys \cite{hds87} designed a approximation scheme for this problem when $m$ is arbitrary.

On the other hand, if $e_{i,k} \in \{0, 1 \}$  for all time intervals, i.e. at any time the machine is either processing a primary job or a routine job but not both,  then our problem reduces to the problem of parallel machine scheduling with availability constraint.  This problem is also  NP-hard and
approximation  algorithms are developed
in  \cite{k98} and \cite{l91}.

For the general problem, $P_m\mid share \mid C_{max}$,  i.e., the sharing ratios  $e_{i,k}$ are arbitrary values in $ (0, 1]$, Fu, et al. \cite{fhz22} showed that there is  no approximation algorithm for the problem
unless $P=NP$. Then  they studied  $P_m\mid share,  e_{i,k} \ge e_0 \mid C_{max}$ and $P_m\mid share, e_{i,k} \ge e_0 \text{ for } i \le m_1 \mid C_{max}$, where $e_0$ is a constant. They analyzed the performance of some classical heuristics for the two problems. Finally, they  developed an approximation scheme for the problem $P_m\mid share,  e_{i,k} \ge e_0 \text{ for } i \le m_1 \mid C_{max}$.

All the above results are for the problems under the traditional  data model where all data can be stored in memory. There is no result for the studied problems under the data stream model where the input data is massive and cannot be read into memory.

While streaming algorithms have been studied in the field of statistics, optimization, and graph algorithms (see surveys by Muthukrishnan [38] and McGregor [37])  in the last twenty years, very little research is conducted in the area of scheduling and operations research in general.
 In \cite{bf12}, Beigel and Fu developed a randomized
streaming approximation scheme for the bin packing problem such that it needs only constant updating time and constant space, and outputs an $(1 + \epsilon)$-approximation in $(1/\epsilon)^{O(1/\epsilon)}$.
In \cite{cv21},  Cormode and Vesel{\'y}
designed a streaming asymptotic $(1+\epsilon)$-approximation algorithms for bin packing and $(d+\epsilon)$-approximation for vector bin packing in $d$ dimensions. For the related vector scheduling problem, they showed how to construct an input summary in space $\widetilde{O}(d^2 \cdot m/\epsilon^2)$ that preserves the optimum value up to a factor of $(2-1/m+\epsilon)$, where $m$ is the number of identical machines.

\subsection{New Contribution}

In this work, we develop the first steaming algorithms for the generalized multitasking shared processing scheduling problems. We assume the machine information is known beforehand, but all the job information including the number of jobs and the processing times are unknown until they are read.
Our streaming algorithms are approximation schemes. In one pass, for any positive constant $\epsilon$, the algorithms return a value that is at most $(1+\epsilon)$ times the optimal value of the makespan and it takes constant updating time to process each job in the stream.
If the jobs can be input in two passes, the algorithms can generate the schedule with constant processing time for each job in the stream.
We show that if an estimate of the maximum processing time can be obtained from priori knowledge, the approximation scheme can be implemented more efficiently. It should be noted that the approximation scheme given by Fu, et al. \cite{fhz22} does not work under the data stream model. So in this paper, we  develop a different approximation scheme.

\section{Streaming Algorithms}

In this section, we will present our streaming algorithms for the multitasking scheduling problem $P_m\mid \text{stream, share, } \  e_{i,k} \ge e_0 \text{ for }i\le m_1 \mid  C_{\max}$.

We first design an approximation scheme for the studied scheduling problem under the tradition model where can be stored in the memory. We then adapt this algorithm for the following three cases, respectively:
(1) the maximum job processing time, $p_{max}$, is given;
(2) an estimate of $p_{max}$ is given;
(3) no information about $p_{max}$ is given. In all these cases, the number of jobs is not known until all jobs are read.

\subsection{An Approximation Scheme under Traditional Data  Model}

In \cite{fhz22}, an approximation scheme has been developed for the studied scheduling problem under traditional data model,  where the main idea is to enumerate all  assignments for large jobs, prune the similar assignments, schedule the small jobs to all the obtained large job assignments, and finally pick the best schedule. This approximation scheme
requires storing all the jobs information and cannot be applied to the data stream model.

So in this subsection we develop a different approximation scheme for the traditional data model but can be adapted to work under the data stream model as well.
The idea of the approximation algorithm is to find the best assignment of large jobs and then generate a single schedule based on this large job assignment.

We first introduce a notation before describing our algorithm.  For any time $t$, we let $A_i(t)$ denote the total amount of processing time of the jobs that can be processed during $(0,t]$ on machine $M_i$. Formally, $A_i(t)={t}\cdot{e_{i,1}}$ if $t\in I_{i,1} = (0, t_{i,1}]$, and $A_i(t)=A_i(t_{i, k})+{(t-t_{i,k})}\cdot{e_{i, k+1}}$ if $t\in I_{i, k+1}=(t_{i,k}, t_{i, k+1}]$.
By the definition of $A_i(t)$, if $N_i$ is a set of jobs assigned to machine $M_i$ and $\sum_{j\in N_i} p_j\le A_i(t)$, then the jobs in $N_i$ can be completed before time $t$ on machine $M_i$.
Let $A(t)$ be the total amount of processing time of the jobs that can be processed during $(0,t]$ on all machines, i.e., $A(t)=\sum_{i=1}^m A_i(t)$.
Our algorithm is outlined as follows.

\medskip
\noindent {\bf Algorithm 1}

\smallskip
{\noindent \bf Input:}
 \begin{itemize}
         \item Parameters $m$, $m_1$,  $e_0$, and $\epsilon$  

          \item The intervals   $(0, t_{i,1}]$, $(t_{i,1}, t_{i,2}]$, $\ldots$ on machine $M_i$,  $1 \le i \le m$,
          and their sharing ratios $e_{i,1}$, $e_{i,2}$, $\ldots$
        \item The number of jobs $n$, and the jobs' processing time $p_j$, $1 \le j \le n$.
     \end{itemize}

{\noindent \bf Output:} A schedule of $n$ jobs

\medskip
{\noindent \bf Steps:}
\begin{enumerate}
 \item Identify the set of large jobs, $JS$, which is constructed below:

  \begin{enumerate}
     \item[(i)] Choose $\gamma_0$ and $N_0$ as follows:
\begin{eqnarray}
2^{\gamma_0 } &=& \ceiling{\frac{m+m_1-1}{\tfrac{\epsilon}{2} \cdot m_1\cdot e_0}} \label{select-gamma0-ineqn}\\
N_0 &=& \ceiling{\frac{m(m+m_1-1)}{\tfrac{\epsilon}{4} \cdot e_0 \cdot m_1}} \label{select-N0-ineqn}
\end{eqnarray}
     \item[(ii)] Let $p_{max}$ be the largest processing time of all the jobs. Define $\gamma_0+2$ continuous intervals, $IH_{-1}=(0, 2^{q_0}]$, $IH_0=(2^{q_0}, 2^{q_0+1}]$, $\ldots$,$IH_k=(2^{q_0+k}, 2^{q_0+k+1}]$, $\ldots$ , $IH_{\gamma_0}=(2^{q_0+\gamma_0}, 2^{q_0+\gamma_0+1}]$, such that $p_{max} \in IH_{\gamma_0}$, i.e., $2^{q_0+\gamma_0}<p_{max}\le 2^{q_0+\gamma_0+1}$.  Correspondingly, partition the jobs in $N = \{1, 2, \cdots, n\}$ into $\gamma_0 + 2$ groups, $H_{-1}$, $H_0$, $H_1$, \ldots, $H_{\gamma_0}$, based on their processing time:  if $p_j \in IH_k$, then   add job $j$  to $H_{k}$.
 \item[(iii)] Let $k_L$, $0 \le k_L \le \gamma_0$ be the largest index of the group that contains at least  $N_0$ jobs; 
 if no such group exists, i.e. all the groups have less than $N_0$ jobs, 
 we let $k_L = -1$.
Let $JS = \cup_{ k > k_L} H_k$ be the set of large jobs, and the remaining jobs are considered as small jobs.
  \end{enumerate}
 \item For each possible assignment of jobs in $JS$, find a time $t$ such that
   \begin{enumerate}
     \item[(i)]
     Both of the following conditions hold for $t$:  (a) $A(t) \ge \sum_{i=1}^n p_i$ and (b) for  $1 \le i \le m$, $A_i(t) \ge P^L_i$  where $P^L_i$ is the total processing time of the large jobs assigned on $M_i$.
     \item[(ii)] At least one of (a) and (b) doesn't hold when $t$ is replaced by ${\tfrac{t}{1+\epsilon/2}}$;
    \end{enumerate}

   \item Among all the assignments of jobs in $JS$, we pick the one that is associated  with the smallest $t$  and do the following:
      \begin{enumerate}
     \item[(i)] Assign the small jobs in any order to the machines so that
     at most  one job finishes at or after $t$ on each machine $M_i$

     \item[(ii)]  Remove the small jobs that finishes after $t$ on machine $M_i$, $i > m_1$ and schedule them to the first $m_1$ machines so that   there are
    at most $\ceiling{\tfrac{m-1}{m_1}}$ jobs finishing after $t$ on $M_i$, $i \le m_1$.
     \end{enumerate}
     \item  Return the obtained schedule.
\end{enumerate}

\begin{lemma}\label{approximation-scheme2-error-ratio-lemma}
  Algorithm~1 outputs a $(1+\epsilon)$-approximation for the scheduling problem $P_m| \text{share, } \  e_{i,k} \ge e_0 \text{ for } i \le m_1 |  C_{\max}$.
\end{lemma}

\begin{proof}
First let us look at the assignment of the jobs in $JS$ that is the same as the optimal schedule. In step 2 of the algorithm, we find the time $t^*$ associated with this large job assignment such that both (a) and (b) hold for $t^*$, but either (a) or (b) doesn't hold when $t^*$ is replaced by ${\tfrac{t^*}{1+\epsilon/2}}$, which means that $\tfrac{t^*}{1+ \epsilon/2} < C_{max}^*$, i.e., $ t^* \le (1  + \tfrac{\epsilon}{2})C_{max}^* $. In step 3 of the algorithm, the large job assignment associated  with the  smallest $t$  is selected, and we have $ t \le t^* \le  (1  + \tfrac{\epsilon}{2})C_{max}^* $.

For the selected large job assignment in step 3 of the algorithm, by condition (b), all large jobs are finished before or at $t$. By condition (a), $A(t) \ge \sum_{i=1}^n p_i$, there must be at most $(m-1)$ small jobs that finish after $t$  in step 3(i), and thus in step 3(ii) we can distribute them  onto the first $m_1$ machines so that at most $\ceiling{\tfrac{m-1}{m_1}}$ small jobs on each of these machines finish after $t$. Hence, the job with the largest completion time must be on one of the first $m_1$ machines.

Let $j$ be the job such that $C_j = C_{max}$, and assume $j$ is scheduled on machine $M_i$, $i \le m_1$. Let $U$ 
be the set of jobs on $M_i$ that finish after $t$. As the sharing ratio is at least $e_0$ on $M_i$,
 $C_j  \le    t + \sum_{ u \in U}  \tfrac{p_u}{e_0}$.
 Note that  all jobs in $U$ are small jobs, i.e.
   $p_u \le 2^{q_0+k_L+1}$  for each job $u \in U$.
From the analysis above, $|U| \le \ceiling{\tfrac{m-1}{m_1}}$.        So,
 \begin{equation} C_{max} = C_j  \le   t+  \ceiling{\frac{m-1}{m_1}} \cdot  \frac{1}{e_0} \cdot 2^{q_0+k_L+1}. \label{algorithm2-makespan-bound}
 \end{equation}

If $k_L > -1$,    there are at least $N_0$ jobs in $H_{k_L}$, and each of which has a processing time in the range of $(2^{q_0+k_L}, 2^{q_0+k_L+1}]$. Therefore, we have $C_{max}^* \ge  \tfrac{ N_0 \cdot 2^{q_0 + k_L}}{m}$, which implies
 $2^{q_0+k_L+1} \le \tfrac{2m C_{max}^*}{N_0}$. Thus, we have
\begin{eqnarray*}
 C_{max} = C_j  & \le  &  t+  \ceiling{\frac{m-1}{m_1}} \cdot  \frac{1}{e_0} \cdot 2^{q_0+k_L+1} \\ 
 & \le & t +   \ceiling{\frac{m-1}{m_1}} \cdot \frac{1}{e_0} \cdot \frac{2 m C_{max}^* }{N_0}\\
  & \le & t + \frac{m+m_1-1}{m_1} \cdot \frac{1}{e_0} \cdot \frac{2 m C_{max}^* }{N_0} \\
 & \le & \left(1+\frac{\epsilon}{2}\right) C_{max}^* + \frac{m+m_1-1}{m_1} \cdot \frac{1}{e_0} \cdot \frac{2 m C_{max}^* }{N_0}  \text{ \hspace{0.1in}  by  Equation  } (\ref{select-N0-ineqn})\\
&    \le & (1 + \epsilon) C_{max}^* \enspace.
\end{eqnarray*}
Otherwise, $k_L = -1$,
using $2^{q_0+\gamma_0}<p_{max}\le 2^{q_0+\gamma_0+1}$,
we have
 $2^{q_0+k_L+1} = 2^{q_0}  \le \tfrac{ p_{max}}{2^{\gamma_0}} \le \tfrac{ C_{max}^*}{2^{\gamma_0}}$. Therefore, we can get
    \begin{eqnarray*}
     C_{max} = C_j & \le    & t + \ceiling{\frac{m-1}{m_1}}  \cdot \frac{1}{e_0} \cdot 2^{q_0+k_L+1} \\
     & =   & t + \ceiling{\frac{m-1}{m_1}}  \cdot \frac{1}{e_0} \cdot 2^{q_0} \\
 & \le & t + \ceiling{\frac{m-1}{m_1}}  \cdot \frac{1}{e_0} \cdot \frac{C_{max}^*}{2^{\gamma_0}} \\
  &  \le & t + \frac{m+m_1-1}{m_1}  \cdot \frac{1}{e_0} \cdot \frac{C_{max}^*}{2^{\gamma_0}} \\
  &  \le & \left(1+\frac{\epsilon}{2}\right) C_{max}^* + \frac{m+m_1-1}{m_1}  \cdot \frac{1}{e_0} \cdot \frac{C_{max}^*}{2^{\gamma_0}}  \text{ \hspace{0.1in}  by  Equation  }   (\ref{select-gamma0-ineqn}) \\
   & \le &  (1 + \epsilon) C_{max}^* \enspace.
    \end{eqnarray*}

So in both cases, we have $C_{max} \le (1+ \epsilon)C_{max}^*$.
\end{proof}

Now we analyze the running time of Algorithm~1.
{\noindent}Let
\begin{eqnarray}
t_1(m,m_1,\tilde{n},\epsilon, e_0)) =   \left( m^{O(\tfrac{m^2}{\epsilon e_0 m_1} \log\tfrac{m}{\epsilon e_0 m_1})}\left( \log (\tfrac{1}{\epsilon}\log \tfrac{m}{ e_0}) \cdot \sum_{1 \le i \le m} \log k_i  \right)\right) \enspace,\label{t2-eqn}
\end{eqnarray}
then we have the following lemma for the running time of Algorithm~1.
\begin{lemma}\label{approximation-scheme2-time-lemma} Let  $\epsilon$ be a real number in $(0,1)$,
Algorithm~1 runs in time $$O\left(n + \tilde{n}  + t_1(m,m_1,\tilde{n},\epsilon, e_0))
\right),$$
which is linear $O(n + \tilde{n})$ when $m$ is constant.
\end{lemma}

\begin{proof}
In Step 1, we find the set of large jobs $JS$ which can be done in $O(n)$ time.
The number of jobs in $JS$ is at most $N_0 (\gamma_0 +1)$.

In step 2, we consider all possible assignments of jobs in $JS$, and there are  $O(m^{O(N_0\gamma_0)})$ of them.
By Equation~(\ref{select-N0-ineqn}), we have
$N_0=O\left(\tfrac{m(m +m_1-1) }{\epsilon e_0 \cdot m_1}\right) = O\left(\tfrac{m^2}{\epsilon e_0 \cdot m_1}\right) $. By Equation~
(\ref{select-gamma0-ineqn}), we have $\gamma_0=O\left(\log{\tfrac{m+m_1-1}{
        \epsilon e_0 \cdot m_1} }\right) =O\left(\log{\tfrac{m}{
        \epsilon e_0 \cdot m_1} }\right).$

Let $P=\sum_{i=1}^n p_i$, then the time $t$ associated with each assignment of jobs in $JS$ is between the lower bound   $LB= P/m$   and the upper bound  $ UB = \tfrac{P}{e_0} = LB \cdot \tfrac{m}{e_0}$. It takes $O(\log (UB-LB))$ time to search the exact smallest $t$ such that both (a) and (b) hold. To speed up the algorithm, instead of searching the exact smallest $t$, we search $t$ with a $(1+\epsilon/2)$-approximation
in step 2(ii).
Specifically, we only need to consider those time points whose values are  $LB\times (1+\tfrac{\epsilon}{2})^x$,  where  $0 \le x \le \log_{1+\epsilon/2} \tfrac{UB}{LB}
= O\left(\tfrac{1}{\epsilon}\log \tfrac{m}{ e_0}  \right)$.
In this way, we can use binary search  to find the corresponding $x$ 
in $O\left( \log (\tfrac{1}{\epsilon}\log \tfrac{m}{ e_0})  \right)$ iterations.

In each iteration of binary search, for the specific $t=LB \cdot (1+\tfrac{\epsilon}{2})^x$, we need to calculate $A_i(t)$, $1\le i \le m$, which is the total amount
of jobs that can be processed by $t$ on machine $M_i$.
To do so, we use binary search to find the interval  $(t_{i,k-1}, t_{i,{k}}]$ such that  $t \in (t_{i,k-1}, t_{i,{k}}]$ in $O(\log  k_i)$ time, then compute
$A_i(t) = A_i(t_{i,k-1}) + (t - t_{i,k-1}) e_{i,k}$, which can be  done in constant time if  $A_i(t_{i,k})$ is known; indeed, we  can pre-calculated $A_i(t_{i,k})$ for all $i$ and $t_{i,k}$ in $O(\tilde{n})$ time. Once  $A_i(t)$ is calculated  for all $i$,    $1 \le i \le m$, we have $A(t) = \sum A_i(t)$. In total,  it takes $O(\sum_{i=1}^{m} \log k_i)$ time to calculate $A_i(t)$ and $A(t)$, and check conditions (a) and (b) for $t$. With the same running time, one can calculate $A_i({\tfrac{t}{1+\epsilon/2}})$ and check conditions (a) and (b) for ${\tfrac{t}{1+\epsilon/2}}$.
%
Therefore, the time of finding $t$ associated with a specific large job assignment is
$\left( \log (\tfrac{1}{\epsilon}\log \tfrac{m}{ e_0}) \cdot \sum_{1 \le i \le m} \log k_i  \right)$. The total time for all assignments
would be
$t_1(m,m_1,\tilde{n},\epsilon, e_0)$, which is given by Equation (\ref{t2-eqn}).

In step 3,  we select the assignment associated  with the smallest $t$ and schedule the small jobs based on this assignment. This can be done in time $O(n + \tilde{n})$.

Adding the time in all steps, we get the total time as stated in the lemma.
\end{proof}

By Lemmas~\ref{approximation-scheme2-error-ratio-lemma} and ~\ref{approximation-scheme2-time-lemma}, we have the following theorem.

\begin{theorem}\label{approximation-scheme-theorem}
Let  $\epsilon$ be a real number in $(0,1)$,
and $m$ be a constant.
Then there is a $(1+\epsilon)$-approximation scheme for  $P_m| \text{ share} \  e_{i,k} \ge e_0 \text{ for } i \le m_1|  C_{\max}$ in time   $O(n+ \tilde{n}) $.
\end{theorem}

Now in the following we will adapt Algorithm~1  so it works under the streaming model where the processing times of the jobs are given as a stream.
In Algorithm~1, we classify jobs as large or small and then process them separately. Whether a job is large or small is determined by the parameters   $\gamma_0$,  $N_0$ and $q_0$ which in turn are determined by  $m$, $m_1$, $e_0$,  $\epsilon$ and $p_{max}$. While   $m$, $m_1$, $e_0$ and $\epsilon$ are  parts of the input,  $p_{max}$ may or may not be. We will first give the streaming algorithm when $p_{max}$ is given as an input.

\subsection{Streaming Algorithm When $p_{max}$ is Given}
Given $P_{max}$ as part of the input, we can  partition the jobs into groups and classify a job as large or  small as in Algorithm~1, but the challenge is that we can not store the processing times of all jobs.
For each group $H_k$, $-1 \le k \le \gamma_0$, we maintain a triple $(n_k, P_k, JS_k)$, where  $n_k$ is the number of jobs in $H_k$, $P_k$ is the total processing time of jobs in $H_k$, 
and $JS_k$ contains the set of jobs in $H_k$ if $n_k<N_0$ and $k \ge 0$; otherwise, $JS_k$ is an empty set. That is, we keep the processing times of the potential large jobs only. And thus we only keep the processing times of at most $N_0(\gamma_0+1)$ large jobs. We update the triples for the groups   as jobs are scanned one by one. Once all the jobs are input, we have the complete information of large jobs and can use step 2 of Algorithm~1 to find the assignment of large jobs associated with the smallest $t$.  Since we are concerned with the approximate value of the optimal makespan, we don't need to schedule the small jobs as in  step 3 of Algorithm 1. Instead we can directly return the approximate value of the optimal makespan as $t+\ceiling{\tfrac{m-1}{m_1}} \cdot \tfrac{1}{e_0} \cdot 2^{q_0+k_L+1}$. The algorithm can be described as follows.

\medskip
\noindent {\bf Algorithm 2}

\smallskip
{\noindent \bf Input:}
 \begin{itemize}
         \item Parameters $m$,    $m_1$,  $e_0$, and $\epsilon$ 

          \item The intervals   $(0, t_{i,1}]$, $(t_{i,1}, t_{i,2}]$, $\ldots$ on $M_i$
          and their sharing ratios $e_{i,1}$, $e_{i,2}$, $\ldots$, respectively ($1 \le i \le m$)
        \item the maximum processing time $p_{max}$
        \item The jobs' processing time $p_j$ (stream input), $1 \le j \le n$.
     \end{itemize}
{\noindent\bf Output:} An approximate value of the optimal makespan.

  \medskip
   {\noindent \bf Steps:}
\begin{enumerate}
 \item Identify the set of large jobs, $JS$.
  \begin{enumerate}
     \item[a.] Choose $\gamma_0$, $N_0$, $q_0$ as in Algorithm 1.

     \item[b.] Read jobs one by one and do the following:

     \begin{enumerate}
    \item if $p_j \le 2^{q_0}$
    \item[] \ \ $k = -1$
    \item[] else
    \item[] \ \ $k = \ceiling{\log_2 p_j} - q_0 -1 $
    \item  $P_{k} = P_k + p_j$
    \item $n_k = n_k +1$
    \item if $k = -1$ or $n_k \ge N_0$
    \item[] \ \  reset $JS_k = \emptyset$
    \item[] else
    \item[] \ \  $JS_k = JS_k \cup \{j\}$
    \end{enumerate}

 \item[c.] $P = \sum_{-1 \le k \le \gamma_0} P_k$
 \item[d.] Let $k_L$, $0 \le k_L \le \gamma_0$, be the largest index of the group such that $n_k \ge N_0$; if $k_L$ doesn't exist, let $K_L = -1$.  Let $JS = \cup_{ k > k_L} JS_k$.
  \end{enumerate}
 \item For each possible assignment of jobs in $JS$, find time $t$ associated with the assignment as in Algorithm 1

   \item Find the  smallest $t$ from the above step, and return the value $t+\ceiling{\tfrac{m-1}{m_1}} \cdot \tfrac{1}{e_0} \cdot 2^{q_0+k_L+1}$.
\end{enumerate}

\begin{theorem}\label{streaming-alg-pmax} Let $\epsilon$ be a real number in $(0,1)$, assume that $p_{max}$ is a given input, then Algorithm 2 is a one-pass streaming algorithm for  $P_m\mid \text{stream, share, } \  e_{i,k} \ge e_0$ for $i\le m_1 \mid  C_{\max}$ and it takes
    \begin{enumerate}
        \item  $O(1)$ time for processing each job in the stream,
        \item   
        $O\left( \tilde{n} + {\frac{m^2}{\epsilon e_0 m_1}}\cdot {\log{\frac{m}{\epsilon e_0}} }\right)$ space, and
        \item
        $O(\tilde{n} + t_1(m,m_1, \tilde{n}, \epsilon, e_0))$ time 
    \end{enumerate}
to find an approximate value of the optimal makespan 
by a $(1+\epsilon)$ factor.
\end{theorem}

\begin{proof} We first show that the returned value is an approximate value of $C_{max}^*$ by a $(1+\epsilon)$ factor.  Since we  keep all the processing times of large jobs, the selected large job assignment and the associated $t$ obtained in Algorithm~2 are exactly the same as Algorithm~1. Following the same argument in Lemma~\ref{approximation-scheme2-error-ratio-lemma},
we can directly return  $(t + \ceiling{\tfrac{m-1}{m_1}} \cdot \tfrac{1}{e_0} \cdot 2^{q_0+k_L+1})$ as an approximate value of the optimal makespan,  which is at most $(1+\epsilon)C_{max}^*$. 

The storage used for the streaming algorithm is mainly $O(\gamma_0)$ triples including at most $O(N_0\gamma_0)$ jobs. The total storage for these triples is
$O(N_0\gamma_0)=O\left({\frac{m^2}{\epsilon e_0 m_1}}\cdot {\log{\frac{m}{\epsilon e_0}} }\right)$.
In addition, we need to store $O(\tilde{n})$ processing sharing intervals.

As seen in step 1b, each job is processed in $O(1)$ time.
After step 1, we get the set of larges jobs, $JS$, including at most $O(N_0 \gamma_0)$ jobs.
After $A(t_{i,k})$ and $A_{i}(t_{i,k})$ are  pre-calculated in $O(\tilde{n})$, the time for step 2 is the same as in Theorem~\ref{approximation-scheme-theorem}, $t_1(m, m_1, \epsilon, e_0,n)$.
\end{proof}

If the jobs can be read  in a second pass, we can return a schedule of all jobs whose makespan is at most $(1 + \epsilon)C_{max}^*$.
Specifically, in the first pass,  we   store the assignment of jobs in $JS$ associated with the minimum $t$ obtained from step 2   as well as the total processing time, $P_i$, of jobs in $JS$ that are assigned to machine $M_i$, $1 \le i \le m$.
In the second pass, we only need to schedule the small jobs. When a job is read, 
if it is a  job in $JS$, 
we don't need to do anything; otherwise
we schedule it to the machine so that at most $\ceiling{\tfrac{m-1}{m_1}}$ small jobs finishing after $t$.

\begin{theorem}
There is a two-pass $(1+\epsilon)$-approximation streaming algorithm for  $P_m\mid \text{stream, share, } \ e_{i,k} \ge e_0$ for $i\le m_1 \mid  C_{\max}$ such that it takes
    \begin{enumerate}
        \item 
        $O(1)$ time to process each job  in both the first pass and the second pass,

        \item   $O\left( \tilde{n} + {\frac{m^2}{\epsilon e_0 m_1}}\cdot {\log{\frac{m}{\epsilon e_0}} }\right)$ space, and
        \item
        $O(\tilde{n} + t_1(m, m_1,\epsilon))$ 
        time
    \end{enumerate}
to return a schedule with a $(1+\epsilon)$ approximation.
\end{theorem}

\subsection{Streaming Algorithm When an Estimate of $P_{max}$ is Given}
Algorithm 2 works when $p_{max}$ is an input. In reality, however, $p_{max}$ may not be obtained accurately without scanning all the jobs.
In many practical scenarios, however, the estimate of $p_{max}$ could be obtained based on priori knowledge. If this is the case, we can modify Algorithm 2 to get the approximate value of the optimal makespan as described below.

Let us assume that we are given  $p_{max}^E$,  an estimate of $p_{max}$,  such that $p_{max} \le p_{max}^E \le \alpha*p_{max}$. And with this input, if we do the same as in Algorithm 2, we would partition the processing time range $(0, p_{max}^E]$ into $\gamma_0+2$ continuous intervals such as $IH_{-1}=(0, 2^{q_0'}]$, $IH_0=(2^{q_0'}, 2^{q_0'+1}]$ $\cdots$, $IH_{\gamma_0}=(2^{q_0'+\gamma_0},2^{q_0'+\gamma_0+1}]$, where $q_0'=\ceiling{\log p_{max}^E}-\gamma_0-1$. Comparing with the intervals obtained with the exact $p_{max}$, we have $q_0 = \ceiling{\log p_{max}} - \gamma_0 -1 \ge \ceiling{\log \tfrac{p_{max}^E}{\alpha}} - \gamma_0 -1 \ge q_0' - \ceiling{\log \alpha}$, so we need to further split $IH_{-1}$ into $\ceiling{\log \alpha}+1$ intervals such as $IH_{-1}=(2^{q_0'-1}, 2^{q_0'}]$, $IH_{-2}=(2^{q_0'-2}, 2^{q_0'-1}]$, $\cdots$, $IH_{-1-\ceiling{\log \alpha}}=(0, 2^{q_0'-\ceiling{\log \alpha}}]$. Correspondingly, we have the job groups $H_{-1-\ceiling{\log \alpha}}$, $\cdots$, $H_{\gamma_0}$. When we scan the jobs one by one, we can add the job to the corresponding group based on its processing time as we did in Algorithm 2. After all the jobs are scanned, we can get the exact $p_{max}$ and only keep $\gamma_0+2$ groups as in Algorithm 2 and other parts of the algorithm will remain the same as Algorithm 2.

\begin{corollary}
\label{streaming-alg-pmax} Let $\epsilon$ be a real number in $(0,1)$, assume that an estimate of $p_{max}$ is a given input, then for  $P_m\mid \text{stream, share, } \  e_{i,k} \ge e_0$ for $i\le m_1 \mid  C_{\max}$ there is a one-pass streaming algorithm to find an approximate value of the optimal makespan by a $(1+\epsilon)$ factor and a two-pass streaming algorithm to find a schedule with this approximate value. The algorithms have $O(1)$ time for processing each job in the stream, $O\left( \tilde{n} + {\frac{m^2}{\epsilon e_0 m_1}}\cdot {\log{\frac{m}{\epsilon e_0}} }\right)$ space usage, and
$O(\tilde{n} + t_1(m,m_1, \tilde{n}, \epsilon, e_0))$ running time.
\end{corollary}

\subsection{Streaming Algorithm When  No Information about $p_{max}$ is Given}

Now we consider the case 
that not only the exact $p_{max}$ but also the estimate of $p_{max}$ is not given. In this case, $\gamma_0$ and $N_0$ can be calculated as before, $q_0$ cannot be determined until all jobs are read, thus  we can not immediately determine which group a job belongs to  as in Algorithm 2. To solve this problem, we need to modify Algorithm 2. When a job is read, we assign it to a group based on the information that we have  so far, as more jobs are read later, we dynamically update the partition of the jobs so that we can maintain the following invariant:
 As in Algorithm 2, there are  $(\gamma_0 + 2)$ groups of jobs, $H_{-1}$, $H_0$, $\ldots$, $H_{\gamma_0}$.   If $p_j \le 2^{q_0}$, then $j \in H_{-1}$; otherwise if $p_j \in ( 2^{q_0 + k}, 2^{q_0 + k +1}]$, $j \in H_k$.   

To implement this efficiently, we use a B-tree (or other balanced search tree) to store information of those non-empty groups $H_k$ for $k \ge 0$ and additionally we maintain the total processing time of jobs in $H_{-1}$, $P_{-1}$. Each group $H_k$ in B-tree is represented as   a quadruple $(\kappa_k, n_k, P_k, JS_k)$, where $n_k$, $P_k$, $JS_k$ are defined the same as in Algorithm~2; $\kappa_k=2^{q_0+k+1}$  is the key representing the processing time range $IH_k=(2^{q_0+k}, 2^{q_0+k+1}]$ of jobs in $H_k$  for $0 \le k \le \gamma_0$.
There are  at most $\gamma_0+1$ quadruples in the B-tree.
Initially, $q_0=0$ and the tree is empty. If a job $j$ has the processing time $p_j \le 2^{q_0+\gamma_0+1}$, 
we  update the quadruple with the key $2^{\ceiling{\log p_j}}$ or insert a new quadruple with the key $2^{\ceiling{\log p_j}}$ if there is no such quadruple in the tree. If   $p_j > 2^{q_0+\gamma_0+1}$, then let $q_0'=\ceiling{\log p_j}-\gamma_0-1$, delete all the quadruples with the key less than or equal to $2^{q_0'}$ and update $P_{-1}$ correspondingly, insert a new quadruple with the key $2^{q_0'+\gamma_0+1}$, and update $q_0=q_0'$.
A detailed description of the algorithm is given below.

\medskip

\noindent {\bf Algorithm 3}

\smallskip
{\noindent \bf Input:}
 \begin{itemize}
         \item Parameters $m$,   $m_1$,  $e_0$, and $\epsilon$  

          \item The intervals   $(0, t_{i,1}]$, $(t_{i,1}, t_{i,2}]$, $\ldots$ on $M_i$
          and their sharing ratios $e_{i,1}$, $e_{i,2}$, $\ldots$, respectively ($1 \le i \le m$)
        \item The jobs' processing time $p_j$ (stream input), $1 \le j \le n$.
     \end{itemize}

{\noindent\bf Output:} An approximate value of the optimal makespan.

  \medskip
   {\noindent \bf Steps:}
\begin{enumerate}
\item Identify the set of large jobs, $JS$.
  \begin{enumerate}
     \item Choose $\gamma_0$ and $N_0$ as in Algorithm 1 and set $q_0=0$ 
     \item Create an empty  B-tree to store the quadruples
     \item Read jobs one by one and do the following:
     \begin{enumerate}
    \item Let $k=\ceiling{\log p_j}-q_0-1$
    \item If the quadruple with the key $2^{k+q_0+1}$ is already in the tree, update as follows:
     \item[] \hspace{0.2in}      $n_k = n_k +1$

     \item[] \hspace{0.2in}      $P_k = P_k + p_j$

    \item[] \hspace{0.2in}     if $n_k \le N_0$,  then $JS_k = JS_k \cup \{j\}$,
    \item[] \hspace{0.2in}      else   reset $JS_k = \emptyset$
\item  Else
    \begin{enumerate}
        \item[] if $p_j \le 2^{q_0}$, let $P_{-1}=P_{-1}+p_j$
        \item[] else if $p_j \le 2^{q_0+\gamma_0+1}$ 
        \item[] \hspace{0.2in} insert a new  quadruple ($2^{q_0 + k +1}, 1, p_j, \{j\}$)
        \item[] else (in this case, $p_j > 2^{q_0+\gamma_0+1}$)
            \item[] \hspace{0.2in} $q_0'= \ceiling{\log p_j} - \gamma_0 -1$

               \item[] \hspace{0.2in} for each $(\kappa_k, n_k, P_k, JS_k)$ in the tree  where  $\kappa_k \le 2^{q_0'}$
            \item[] \hspace{0.4in}   $P_{-1} = P_{-1} + P_k$,

            \item[] \hspace{0.4in} delete the quadruple $(\kappa_k, n_k, P_k, JS_k)$  from the tree.

            \item[] \hspace{0.2in} insert   quadruple ($2^{q_0'+\gamma_0+1}, 1, p_j, \{j\}$) in the tree;
            \item[] \hspace{0.2in} update $q_0=q_0'$.
    \end{enumerate}
    \end{enumerate}
 \item  $P  = \sum_{-1 \le k \le \gamma_0} P_k$.
 \item  Let $k_L$, $0 \le k_L \le \gamma_0$, be the largest index of the group such that $n_k \ge N_0$, if $k_L$ doesn't exist, let $K_L = -1$.  Let $JS = \cup_{ k > k_L} JS_k$.

  \end{enumerate}
 \item For each possible assignment of jobs in $JS$, find the time $t$ associated with it as in Algorithm 1

   \item Find the smallest $t$ from previous step, and return the value $t+\ceiling{\tfrac{m-1}{m_1}} \cdot \tfrac{1}{e_0} \cdot 2^{q_0+k_L+1}$.
\end{enumerate}

\begin{theorem}\label{approximation-scheme-thm2} Let     $\epsilon$ be a real number in $(0,1)$.
Then Algorithm 3 is a one-pass $(1+\epsilon)$-approximation streaming algorithm for $P_m\mid \text{stream, share,} \  e_{i,k} \ge e_0 $ for $i\le m_1 \mid C_{\max}$ such that it takes
    \begin{enumerate}
        \item  
        $O(1)$  time to process each job in the stream,
        \item
        $O\left( \tilde{n} + {\frac{m^2}{\epsilon e_0 m_1}}\cdot {\log{\frac{m}{\epsilon e_0}} }\right)$ space, and
        \item 
        $O(\tilde{n} + t_1(m,m_1, \tilde{n}, e_0, \epsilon))$
        time
    \end{enumerate}
    to find an
        approximation of  $C_{max}^*$ by a factor of  $(1+\epsilon)$.
\end{theorem}

\begin{proof}
The proof is similar to that of  Theorem~\ref{streaming-alg-pmax}, so  we only discuss the difference - the updating time for   each job.

We use a B-tree (or other balanced search tree) to store the job  groups where each group is represented as a quadruple. At any time, there are   at most $\gamma_0+1 = O(\log \tfrac{m}{\epsilon e_0})$ quadruples/keys in tree.

For each job $j$, we perform a search operation, and maybe insert or delete. Since the number of keys/quadruples in the B-tree is $O(\gamma_0)$,  all these operations  can be done in 
$O(\log \gamma_0)=O(\log  {\log{\tfrac{m}{\epsilon e_0}} })$ time, which is O(1) since $m$ is a constant.
\end{proof}

Similarly, we can find the approximation schedule in two passes.

\begin{theorem}
There is a two-pass $(1+\epsilon)$-approximation streaming algorithm for  $P_m\mid \text{stream, share} \  (e_{i,k} \ge e_0)$ for $i\le m_1 \mid  C_{\max}$ such that it takes
    \begin{enumerate}
        \item
        $O(1)$ time to process each job in the stream,
        \item   $O\left( \tilde{n} + {\frac{m^2}{\epsilon e_0 m_1}}\cdot {\log{\frac{m}{\epsilon e_0}} }\right)$ space, and
        \item 
        $O(\tilde{n} + t_1(m, m_1,\epsilon))$ 
        time
    \end{enumerate}
to find a $(1+\epsilon)$ approximation of the optimal makespan  after receiving all jobs in the stream in the first pass, and
         $O(1)$ time for each job in the second pass to return a schedule for all jobs.
\end{theorem}

\section{Conclusions}

In this paper, we studied the multitasking scheduling problem with shared processing under the data stream model. There are multiple machines with sharing ratios varying from one interval to another, and we allow the sharing ratios on some machines have a   constant lower bound. The goal is to minimize the makespan.

We designed the first streaming approximation schemes for the problem where the processing times of the jobs are input as a stream and no prior information about the number of jobs is required.
This work not only provides an algorithmic big data solution for our studied scheduling problem, but also leads to one future research direction for the area of scheduling. The classical scheduling literature contains
a large number of problems that remain to be studied under the data stream model presented here.

For our studied problems, it is also interesting to design streaming algorithms for other performance criteria including total completion time, maximum tardiness, and other machine environments such as uniform machines, flowshop, etc. 

\begin{thebibliography}{29}
\providecommand{\natexlab}[1]{#1}
\providecommand{\url}[1]{\texttt{#1}}
\expandafter\ifx\csname urlstyle\endcsname\relax
  \providecommand{\doi}[1]{doi: #1}\else
  \providecommand{\doi}{doi: \begingroup \urlstyle{rm}\Url}\fi

\bibitem[Adiri and Yehudai(1987)]{ay87}
I.~Adiri and Z.~Yehudai.
\newblock Scheduling on machines with variable service rates.
\newblock \emph{Computers \& Operations Research}, 14\penalty0 (4):\penalty0
  289–297, 1987.
\newblock \doi{10.1016/0305-0548(87)90066-9}.

\bibitem[Alon et~al.(1999)Alon, Matias, and Szegedy]{ams99}
Noga Alon, Yossi Matias, and Mario Szegedy.
\newblock The space complexity of approximating the frequency moments.
\newblock \emph{Journal of Computer and System Sciences}, 58\penalty0
  (1):\penalty0 137--147, 1999.
\newblock \doi{https://doi.org/10.1006/jcss.1997.1545}.

\bibitem[Baker and Nuttle(1980)]{bn80}
Kenneth~R. Baker and Henry L.~W. Nuttle.
\newblock Sequencing independent jobs with a single resource.
\newblock \emph{Naval Research Logistics Quarterly}, 27:\penalty0 499--510,
  1980.

\bibitem[Beigel and Fu(2012)]{bf12}
Richard Beigel and Bin Fu.
\newblock A dense hierarchy of sublinear time approximation schemes for bin
  packing.
\newblock \emph{Electronic Colloquium on Computational Complexity},
  18:\penalty0 28, 2012.

\bibitem[Cormode and Vesel{\'y}(2021)]{cv21}
Graham Cormode and Pavel Vesel{\'y}.
\newblock Streaming algorithms for bin packing and vector scheduling.
\newblock \emph{Theory of Computing Systems}, 65:\penalty0 916--942, 2021.
\newblock \doi{10.1007/s00224-020-10011-y}.

\bibitem[Coviello et~al.(2014)Coviello, Ichino, and Persico]{cip14}
Decio Coviello, Andrea Ichino, and Nicola Persico.
\newblock Time allocation and task juggling.
\newblock \emph{American Economic Review}, 104\penalty0 (2):\penalty0 609--23,
  2014.
\newblock \doi{10.1257/aer.104.2.609}.

\bibitem[Dereniowski and Kubiak(2017)]{dk17}
Dariusz Dereniowski and Wieslaw Kubiak.
\newblock Shared multi-processor scheduling.
\newblock \emph{European Journal of Operational Research}, 261:\penalty0
  503--514, 2017.

\bibitem[Dereniowski and Kubiak(2020)]{dk20}
Dariusz Dereniowski and Wieslaw Kubiak.
\newblock Shared processor scheduling of multiprocessor jobs.
\newblock \emph{European Journal of Operational Research}, 282:\penalty0
  464--477, 2020.

\bibitem[Flajolet and {Nigel Martin}(1985)]{fm85}
Philippe Flajolet and G.~{Nigel Martin}.
\newblock Probabilistic counting algorithms for data base applications.
\newblock \emph{Journal of Computer and System Sciences}, 31\penalty0
  (2):\penalty0 182--209, 1985.
\newblock \doi{https://doi.org/10.1016/0022-0000(85)90041-8}.

\bibitem[Fu et~al.(2022)Fu, Huo, and Zhao]{fhz22}
Bin Fu, Yumei Huo, and Hairong Zhao.
\newblock Multitasking scheduling with shared processing, 2022.
\newblock Manuscript under review.

\bibitem[Gonz{\'a}lez and Mark(2005)]{gm05}
V{\'i}ctor~M. Gonz{\'a}lez and Gloria Mark.
\newblock Managing currents of work: Multi-tasking among multiple
  collaborations.
\newblock In \emph{European Conference of Computer-supported Cooperative Work},
  2005.

\bibitem[Graham et~al.(1979)Graham, Lawler, Lenstra, and Kan]{gllr79}
R.L. Graham, E.L. Lawler, J.K. Lenstra, and A.H.G.Rinnooy Kan.
\newblock Optimization and approximation in deterministic sequencing and
  scheduling: a survey.
\newblock In P.L. Hammer, E.L. Johnson, and B.H. Korte, editors, \emph{Discrete
  Optimization II}, volume~5 of \emph{Annals of Discrete Mathematics}, pages
  287--326. Elsevier, 1979.
\newblock \doi{https://doi.org/10.1016/S0167-5060(08)70356-X}.

\bibitem[Graham(1966)]{g66}
Ronald~L. Graham.
\newblock Bounds for certain multiprocessing anomalies.
\newblock \emph{Bell System Technical Journal}, 45:\penalty0 1563--1581, 1966.

\bibitem[Graham(1969)]{g69}
Ronald~L. Graham.
\newblock Bounds on multiprocessing timing anomalies.
\newblock \emph{SIAM Journal of Applied Mathematics}, 17:\penalty0 416--429,
  1969.

\bibitem[Hall et~al.(2015)Hall, Leung, and lun Li]{hll15}
Nicholas~G. Hall, Joseph Y.-T. Leung, and Chung lun Li.
\newblock The effects of multitasking on operations scheduling.
\newblock \emph{Production and Operations Management}, 24:\penalty0 1248--1265,
  2015.

\bibitem[Hall et~al.(2016)Hall, Leung, and lun Li]{hll16}
Nicholas~G. Hall, Joseph Y.-T. Leung, and Chung lun Li.
\newblock Multitasking via alternate and shared processing: Algorithms and
  complexity.
\newblock \emph{Discrete Applied Mathematics}, 208:\penalty0 41--58, 2016.

\bibitem[Hezarkhani and Kubiak(2015)]{hk15}
Behzad Hezarkhani and Wieslaw Kubiak.
\newblock Decentralized subcontractor scheduling with divisible jobs.
\newblock \emph{Journal of Scheduling}, 18:\penalty0 497--511, 2015.

\bibitem[Hirayama and Kijima(1992)]{hk92}
Tetsuji Hirayama and Masaaki Kijima.
\newblock Single machine scheduling problem when the machine capacity varies
  stochastically.
\newblock \emph{Operations Research}, 40:\penalty0 376--383, 1992.

\bibitem[Hochbaum and Shmoys(1987)]{hds87}
Dorit~S. Hochbaum and David~B. Shmoys.
\newblock Using dual approximation algorithms for scheduling problems
  theoretical and practical results.
\newblock \emph{Journal of the ACM}, 34\penalty0 (1):\penalty0 144–162, 1987.

\bibitem[Horowitz and Sahni(1976)]{hs76}
Ellis Horowitz and Sartaj Sahni.
\newblock Exact and approximate algorithms for scheduling nonidentical
  processors.
\newblock \emph{Journal of the ACM}, 23:\penalty0 317--327, 1976.

\bibitem[Kellerer(1998)]{k98}
Hans Kellerer.
\newblock Algorithms for multiprocessor scheduling with machine release times.
\newblock \emph{IIE Transactions}, 30:\penalty0 991--999, 1998.

\bibitem[Lee(1991)]{l91}
Chung-Yee Lee.
\newblock Parallel machines scheduling with nonsimultaneous machine available
  time.
\newblock \emph{Discrete Applied Mathematics}, 30:\penalty0 53--61, 1991.

\bibitem[Munro and Paterson(1980)]{mp78}
J.I. Munro and M.S. Paterson.
\newblock Selection and sorting with limited storage.
\newblock \emph{Theoretical Computer Science}, 12\penalty0 (3):\penalty0
  315--323, 1980.
\newblock \doi{https://doi.org/10.1016/0304-3975(80)90061-4}.

\bibitem[Muthukrishnan(2005)]{m05}
S.~Muthukrishnan.
\newblock Data streams: Algorithms and applications.
\newblock \emph{Foundations and Trends in Theoretical Computer Science},
  1\penalty0 (2):\penalty0 117–236, aug 2005.
\newblock \doi{10.1561/0400000002}.

\bibitem[O’Leary et~al.(2006)O’Leary, Liebovitz, and Baker]{olb06}
Kevin~J. O’Leary, David~M. Liebovitz, and David~W. Baker.
\newblock How hospitalists spend their time: insights on efficiency and safety.
\newblock \emph{Journal of hospital medicine}, 1:\penalty0 88--93, 2006.

\bibitem[Sanbonmatsu et~al.(2013)Sanbonmatsu, Strayer, Medeiros-Ward, and
  Watson]{ssmw13}
David~M. Sanbonmatsu, David~L. Strayer, Nathan Medeiros-Ward, and Jason Watson.
\newblock Who multi-tasks and why? multi-tasking ability, perceived
  multi-tasking ability, impulsivity, and sensation seeking.
\newblock \emph{PLoS ONE}, 8\penalty0 (1), 2013.

\bibitem[Sum and Ho(2015)]{sh15}
John Sum and Kevin Ho.
\newblock Analysis on the effect of multitasking.
\newblock In \emph{2015 IEEE International Conference on Systems, Man, and
  Cybernetics}, pages 204--209, 2015.
\newblock \doi{10.1109/SMC.2015.48}.

\bibitem[Vega et~al.(2008)Vega, McCracken, Nass, and Labs]{vmn08}
Vanessa Vega, Kristle~I. McCracken, Clifford Nass, and Lumos Labs.
\newblock Multitasking effects on visual working memory, working memory and
  executive control, 2008.
\newblock Presentation at annual Meeting of the International Communication
  Association.

\bibitem[Zhu et~al.(2017)Zhu, Zheng, and Chu]{zzc17}
Zhanguo Zhu, Feifeng Zheng, and Chengbin Chu.
\newblock Multitasking scheduling problems with a rate-modifying activity.
\newblock \emph{International Journal of Production Research}, 55:\penalty0 296
  -- 312, 2017.

\end{thebibliography}

\end{document}